\documentclass{article}
\usepackage{epstopdf} 
\usepackage{amsmath}
\usepackage{amsfonts}
\usepackage{amssymb}
\usepackage{graphicx} 
\usepackage{latexsym}
\usepackage{algorithmic}
\usepackage{colortbl}
\usepackage{url}

\newtheorem{theorem}{Theorem}

\newtheorem{conjecture}[theorem]{Conjecture}

\newtheorem{problem}[theorem]{Problem}

\newenvironment{proof}[1][Proof]{\noindent\textbf{#1.} }{\ \rule{0.5em}{0.5em}}
\begin{document}

\title{On the synchronization of planar automata}
\author{J. Andres Montoya, Christian Nolasco \\
Universidad Nacional de Colombia, Bogota, Colombia.}
\maketitle

\begin{abstract}
Planar automata seems to be representative of the synchronizing behavior of
deterministic finite state automata. We conjecture that \v{C}erny's
conjecture holds true, if and only if, it holds true for planar automata. In
this paper we have gathered some evidence concerning this conjecture. This
evidence\ amounts to show that the class of planar automata is
representative of the algorithmic hardness of synchronization
\end{abstract}

This work is related to the synchronization of deterministic finite state
automata (DFAs, for short).

Let $\mathcal{M}$ be a DFA, and let $\Sigma _{\mathcal{M}}$ be its input
alphabet, we use the symbol $\Sigma _{\mathcal{M}}^{\ast }$ to denote the
set of finite strings over the alphabet $\Sigma _{\mathcal{M}}$. The
function $\widehat{\delta _{\mathcal{M}}}:\Sigma _{\mathcal{M}}^{\ast
}\times Q_{\mathcal{M}}\rightarrow Q_{\mathcal{M}}$ is defined by the
equation:%
\begin{equation*}
\widehat{\delta _{\mathcal{M}}}\left( w_{1}...w_{n},q\right) =\delta _{%
\mathcal{M}}\left( w_{n},\widehat{\delta _{\mathcal{M}}}\left(
w_{1}...w_{n-1},q\right) \right) ,
\end{equation*}%
where $\delta _{\mathcal{M}}$ is the transition function of $\mathcal{M}$.

A \textit{synchronizing string} (\textit{reset word}) for $\mathcal{M}$, is
a string $w\in \Sigma _{\mathcal{M}}^{\ast }$ such that for all $p,q\in Q_{%
\mathcal{M}},$ the equality 
\begin{equation*}
\widehat{\delta _{\mathcal{M}}}\left( w,p\right) =\widehat{\delta _{\mathcal{%
M}}}\left( w,q\right)
\end{equation*}%
holds

We say that automaton $\mathcal{M}$ is \textit{synchronizing,} if and only
if, there exists a synchronizing string for $\mathcal{M}$. Let $\mathcal{M}$
be a synchronizing automaton, its \textit{minimal reset length}, denoted by $%
rl_{\mathcal{M}}$, is the length of its minimal synchronizing strings. It is
easy to prove that $rl_{\mathcal{M}}\in O\left( \left\vert Q_{\mathcal{M}%
}\right\vert ^{3}\right) .$ \v{C}erny \cite{C} conjectured that $rl_{%
\mathcal{M}}\leq \left( \left\vert Q_{\mathcal{M}}\right\vert -1\right)
^{2}. $ This conjecture is called \textit{\v{C}erny's Conjecture}, and it is
considered the most important open problem in the combinatorial theory of
finite state automata.

\textbf{The universality conjecture for planar automata. }It is well known
that if \v{C}erny's conjecture holds true for strongly connected automata,
then it holds true for all the deterministic finite state automata.
Therefore, we say that the class of strongly connected automata is \textit{%
universal}. We conjecture that the same is true of the class of planar
automata.

Let us discuss some of the facts that led us to formulate the \textit{%
universality conjecture }for planar automata\textit{.}

We are interested in some algorithmic problems related to
DFA-synchronization. It happens that the algorithmic complexity of those
problems is well understood, and there are many deep results characterizing
their intrinsic hardness \cite{E}, \cite{GS}, \cite{OU}. It can be noticed
that all those hardness proofs work well for the planar restrictions of the
problems. Then, we have that the planar restrictions of those problems are
as hard as the unrestricted versions. It means that the class of planar
automata is an \textit{universal class} with respect to the algorithmic
hardness of synchronization.

It can be checked that all the sequences of slowly synchronizing automata
registered in the literature are sequences of planar automata (see reference 
\cite{AGV}). At his point, it is important to remark that it is fairly easy
to transform a planar sequence of slowly synchronizing automata into a
nonplanar sequence with the same synchronizing behavior, and then it follows
that there exist nonplanar sequences of slowly synchronizing automata.
However, it seems that all the sequences of slowly synchronizing automata
can be obtained this way: By locally perturbing a sequence of slowly
synchronizing planar automata. Notice that if the last assertion is true,
all the automata that could refute \v{C}erny's conjecture are \textit{%
essentially planar}.

The above two observations are the origin of the conjecture. By the way we
have found some additional evidence in favor of it. The goal of this work is
to discuss\ in more detail those old and new facts.

\textbf{Organization of the work and contributions. }This work is organized
into two sections, in section 1 we characterize the algorithmic hardness of
some synchronization problems related to planar automata and we show that
those problems are as hard as the nonplanar versions. Those results amount
to show that the class of planar automata is universal with respect to the
hardness of synchronization. We finish in section 2 with some concluding
remarks.

\bigskip

\section{On the algorithmic hardness of synchronizing planar automata}

We investigate the synchronization of finite state automata focussing on the
class of deterministic planar automata. A finite state automaton is planar,
if and only if, its transition digraph is planar. Planar automata have been
previously studied, and it is known that there are regular languages which
cannot be recognized by deterministic planar automata \cite{Book}. This last
fact indicates that the class of planar automata is not universal with
respect to the recognition power of finite automata. However, we conjecture
that this restricted class is universal with respect to the hardness of
synchronization. This conjecture motivates us to study the synchronization
of planar automata. To begin with, we study the algorithmic complexity of
some synchronization problems for planar and nonplanar automata.

\begin{problem}
($\mathbf{Synch}\left[ P\right] \mathbf{:}$ optimal synchronization of
planar automata)

\begin{itemize}
\item Input: $\left( \mathcal{M},k\right) $, where $\mathcal{M}$ is a
synchronizing planar automaton and $k$ is a positive integer.

\item Problem: Decide if there exists a synchronizing string for $\mathcal{M}
$ whose length is upperbounded by $k.$
\end{itemize}
\end{problem}

\begin{theorem}
The problem $\mathbf{Synch}\left[ P\right] $ is NP complete.
\end{theorem}

\begin{proof}
Eppstein \cite{E} proved there exists a ptime algorithm, which, on input $%
\alpha $ (where $\alpha $ is a CNF) computes a pair $\left( \mathcal{M}%
_{\alpha },k_{\alpha }\right) $, such that $\mathcal{M}_{\alpha }$ is a
synchronizing automaton satisfying the following two conditions:

\begin{enumerate}
\item If $\alpha $ is satisfiable there exists a reset word for $\mathcal{M}%
_{\alpha }$, whose length is upperbounded by $k_{\alpha }.$

\item If $\alpha $ is not satisfiable the length of the minimal reset words
for $\mathcal{M}_{\alpha }$ is equal to $k_{\alpha }+1.$
\end{enumerate}

One can easily check that for all $\alpha ,$ the automaton $\mathcal{M}%
_{\alpha }$ is planar. Thus, we have that set of outputs of Eppstein's
reduction is included in the class of planar synchronizing automata, and it
implies that SAT is ptime reducible to $Synch\left[ P\right] $. Thus, we
have that $Synch\left[ P\right] $ is NP hard. It is easy to check that $Synch%
\left[ P\right] $ belongs to NP
\end{proof}

It was fairly easy to prove that $Synch\left[ P\right] $ is NP complete, we
just noticed that all the gadgets used in Eppstein's proof are planar. A
similar fact will happen more than once in this work: We get the hardness
result for planar automata by noticing that the proof for general automata
works verbatim in the planar framework.

How hard is the problem of approximating the minimal reset length of an
automaton? First, we observe that Eppstein's greedy algorithm \cite{E} is a
ptime approximation algorithm of ratio $O\left( n\right) $, and it is clear
that it must work as well when one restricts its execution to planar
synchronizing automata. It is natural to ask: Which is the best
approximation ratio that can be achieved in polynomial time?

\begin{theorem}
Given $\varepsilon >0,$ it is NP hard to approximate the minimal reset
length of planar automata within the ratio $O\left( n^{1-\varepsilon
}\right) .$
\end{theorem}

\begin{proof}
Gawrychowski and Straszak \cite{GS} proved that for all $\varepsilon >0,$ it
is NP hard to approximate the minimal reset length of general synchronizing
automata within the ratio $O\left( n^{1-\varepsilon }\right) .$ Once again,
it is enough to check that the proof of Gawrychowski and Straszak works for
planar automata.
\end{proof}

Thus, we have that the best approximation ratio that can be achieved in
polynomial time is the ratio $O\left( n\right) ,$ which is achieved by
Eppstein algorithm. Moreover, the claim is true for planar and general
synchronizing automata.

It is worth to remark that the computation of minimal reset lengths is not a
typical NP computation, given that minimality corresponds to an universal
assertion instead of an existential assertion. Then, it cannot be said that
the NP completeness of $synch$ characterizes the intrinsic hardness of
computing minimal reset lengths and minimal synchronizing strings. This
observation motivates the study of a second algorithmic problem, denoted by $%
ESynch\left[ P\right] $ and defined by:

\begin{problem}
($ESynch\left[ P\right] :$ Deciding minimal reset length)

\begin{itemize}
\item \textit{Input: }$\left( \mathcal{M},k\right) $\textit{, where }$%
\mathcal{M}$\textit{\ is a synchronizing planar automaton and }$k$\textit{\
is a positive integer.}

\item \textit{Problem: Decide if the minimal reset length of }$\mathcal{M}$%
\textit{\ is equal to }$k.$
\end{itemize}
\end{problem}

Let DP be the closure under finite intersections of the class NP$\cup $%
co-NP, we prove that $ESynch\left[ P\right] $ is complete for DP. Olschewski
and Ummels proved that $ESynch\left[ P\right] $ is complete for DP (see
reference \cite{OU}). Our result is, once again, an easy consequence of the
nonplanar result (with its proof), but this time we have to work a little
bit.

\begin{theorem}
$ESynch\left[ P\right] $ is complete for DP.
\end{theorem}

As remarked before Olschewski and Ummels proved that $ESynch$ is complete
for DP$.$ To this end, they exhibited a ptime reduction of the problem
SAT-UNSAT in the problem $ESynch.$ Recall that SAT-UNSAT is the problem
defined by:

\begin{proof}

\bigskip

\begin{itemize}
\item \textit{Input: }$\left( \alpha ,\beta \right) $\textit{, where }$%
\alpha $\textit{\ and }$\beta $\textit{\ are boolean formulas in conjunctive
normal form.}

\item \textit{Problem: Decide if }$\alpha $\textit{\ is satisfiable and }$%
\beta $\textit{\ is unsatisfiable.}
\end{itemize}
\bigskip
\begin{figure}[h]
\begin{center}
\includegraphics{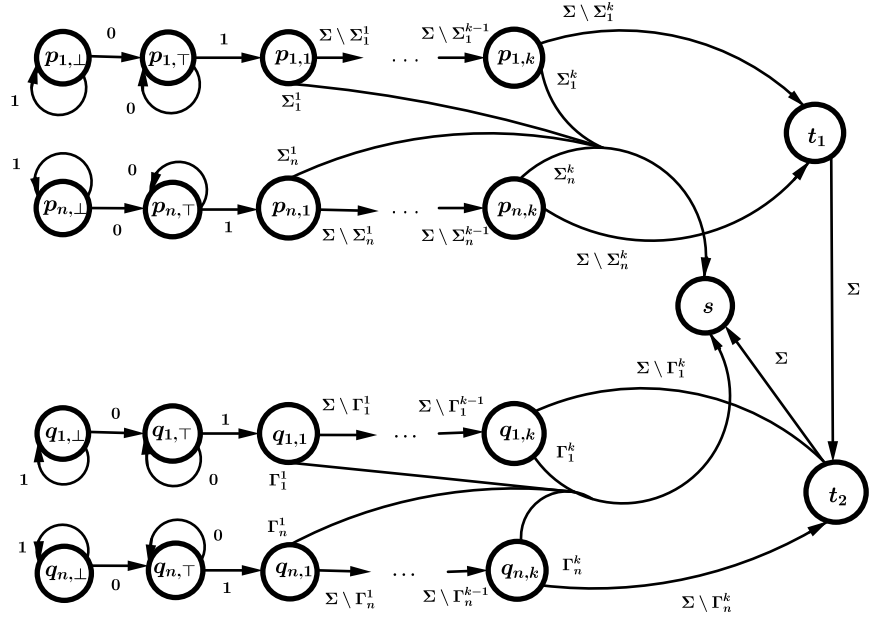}

\end{center}
\caption{Olschewski-Ummels construction}
\label{figura1:Figure 1}
\end{figure}
It is known that the later problem is DP complete, and hence the reduction
of Olschewski and Ummels suffices. Given $\alpha \left(
X_{1},...,X_{k}\right) $ and $\beta \left( X_{1},...,X_{k}\right) $, the
aforementioned reduction allows one to compute a pair $\left( \mathcal{M}%
_{\alpha \beta },k+3\right) $ such that:

\begin{itemize}
\item If $\alpha $ and $\beta $ are satisfiable, then the minimal reset
length of $\mathcal{M}_{\alpha \beta }$ is equal to $k+2.$

\item If $\alpha $ is satisfiable and $\beta $ is unsatisfiable, then the
minimal reset length of $\mathcal{M}$ is equal to $k+3.$

\item If $\alpha $ is unsatisfiable, then the minimal reset length of $%
\mathcal{M}_{\alpha \beta }$ is equal to $k+4.$
\end{itemize}

The automaton $\mathcal{M}_{\alpha \beta }$ is defined as follows:

We suppose, without loss of generality, that $\alpha $ and $\beta $ have the
same number $n$ of clauses, and no propositional variable occurs in both $%
\alpha $ and $\beta $. Let $\alpha =C_{1}\wedge \cdots \wedge C_{n}$ and $%
\beta =D_{1}\wedge \cdots \wedge D_{n}.$ The automaton $\mathcal{M}_{\alpha
\beta }$ consists of the states $s,t_{1},t_{2},p_{ij}$ and $q_{ij}$, where $%
i\in \left\{ 1,...,n\right\} $ and $j\in \left\{ \top ,\bot ,1,...,k\right\}
;$ the transitions are depicted in figure 1. An edge from $p$ to $q$
labelled with $\Delta \subseteq \Sigma $ has the meaning that $\delta \left(
p,a\right) =q$ for all $a\in \Delta .$ The sets $\Sigma _{i}^{j}$ are
defined by $0\in \Sigma _{i}^{j}\Leftrightarrow \lnot X_{j}\in C_{i}$ and $%
1\in \Sigma _{i}^{j}\Leftrightarrow X_{j}\in C_{i},$ and the sets $\Gamma
_{i}^{j}$ are defined by $0\in \Gamma _{i}^{j}\Leftrightarrow \lnot X_{j}\in
D_{i}$ and $1\in \Gamma _{i}^{j}\Leftrightarrow X_{j}\in D_{i}.$


It can be observed that $\mathcal{M}_{\alpha \beta }$ is not always a planar
automaton. However, such a construction can be slightly modified in order to
obtain a planar automaton $\mathcal{M}_{\alpha \beta }^{\ast }$ satisfying
the above three conditions. A possible modification consists in eliminating
the nodes $t_{1}$ and $t_{2}$ and replacing them with the set of nodes%
\begin{equation*}
\left\{ p_{i,,j}:i\leq n,\text{ }j=k+1,k+2\right\} \cup \left\{
q_{i,,j}:i\leq n\text{, }j=k+1\right\} .
\end{equation*}

Moreover, for each $a\in \Sigma $ we set%
\begin{equation*}
\delta \left( q,a\right) =\left\{ 
\begin{array}{c}
p_{i,,j+1}\text{ if }q=p_{i,,j}\text{ and }j=k,k+1 \\ 
s\text{, if }q=p_{i,k+2} \\ 
q_{i,,j+1}\text{ if }q=q_{i,,j}\text{ and }j=k \\ 
s\text{, if }q=q_{i,k+1}%
\end{array}%
\right.
\end{equation*}

It is easy to check that the automaton so defined, which we denote with
symbol $\mathcal{M}_{\alpha \beta }^{\ast }$, is a planar automaton
satisfying the same three conditions satisfied by $\mathcal{M}_{\alpha \beta
}$. Thus, we have a ptime reduction of SAT-UNSAT in the problem $ESynch\left[
P\right] ,$ and hence we can conclude that the later problem is complete for
DP.
\end{proof}

Olschewski and Ummels also proved that the problem of computing the minimal
reset length of a given synchronizing automaton is complete for the class $%
FP^{NP\left[ \log \left( n\right) \right] }$ \cite{OU}. It can be checked
that their proof can be used (verbatim) to show that computing the minimal
reset length of planar automata is complete for the same class of functions.

\subsection{Synchronizing small sets of states}

In this section we add a further hardness result to the above list. This
time we have to work hard, because we have to prove first the corresponding
nonplanar result, which characterizes the parameterized complexity of 
\textit{subset synchronization. }We refer the reader to \cite{FG} for a
pedagogical introduction to the basics of parameterized complexity.

Given an automaton $\mathcal{M}$, and given $q_{1},...,q_{k}\in Q_{\mathcal{M%
}}$, a synchronizing string for those $k$ states, is a string $w,$ such that
for all $i,j\leq k$, the equality 
\begin{equation*}
\widehat{\delta _{\mathcal{M}}}\left( w,q_{i}\right) =\widehat{\delta _{%
\mathcal{M}}}\left( w,q_{j}\right)
\end{equation*}%
holds. In the later case we say that $w$ synchronizes the subset $\left\{
q_{1},...,q_{k}\right\} .$

We think that subset synchronization is a powerful concept that allows one
to model some interesting discrete dynamics. Suppose, for instance, that one
has a troop of agents scattered over a territory, and that he wants to
broadcast an instruction, the same one for all the agents, which must lead
the agents to a common site on the territory. If the territory is the
transition digraph of a synchronizing automaton, and he does not know the
initial locations of\ the agents, then he must broadcast a reset word for
the underlying automaton. On the other hand, if he knows the initial
locations $q_{1},...,q_{k}$, then he must broadcast a synchronizing string
for these $k$ states.

A second dynamics refers a class of games on checkerboards. Suppose one has
a set of tokens scattered over checkerboard. Each time he chooses an
available action, it determines the way in which he must move each one of
his tokens. Tokens that arrive to the same site are stacked. The goal is to
gather all the tokens into a single stack using as few actions as possible.
Notice that if the transition digraph defined by the checkerboard together
with the set of allowed actions defines the transition digraph of an
automaton, say $\mathcal{M}$, then the goal of this game corresponds to
synchronize the states of $\mathcal{M}$ that were occupied by tokens at time 
$0.$

Notice that, in the above two situations, one is interested in computing
minimal synchronizing strings for the set of occupied states. We consider
that planar automata are natural scenarios for the most representative
instances of those two problems: On one hand, it can be argued that planar
digraphs, which are the discrete versions of the surfaces of genus $0,$ are
natural discrete models of territories. And, on the other hand, most
checkerboards are planar.

First, we consider the case where the number of states to be synchronized is
upperbounded by a fixed constant. We use the symbol $k$\textbf{-}$Synch$ to
denote the problem.

\begin{problem}
($k$-$synch:$ Synchronization of $k$-tuples$)$

\begin{itemize}
\item \textit{Input: }$\left( \mathcal{M},\left\{ q_{1},...,q_{r}\right\}
\right) ,$ \textit{where\ }$\mathcal{M}$ \textit{is a synchronizing
automaton, }$q_{1},...,q_{r}\in Q_{\mathcal{M}}$ \textit{and} $r\leq k.$

\item \textit{Problem: Compute a synchronizing string of minimal length for
the states }$q_{1},...,q_{r}.$
\end{itemize}
\end{problem}

Given $k\geq 2,$ it is easy to show that the problem $k$-$Synch$ can be
solved in polynomial time.

\begin{theorem}
$k$-$Synch$ can be solved in time $O\left( \left\vert \mathcal{M}\right\vert
^{2k}\right) .$
\end{theorem}

\begin{proof}
Let $\left( \mathcal{M},\left\{ q_{1},...,q_{r}\right\} \right) $ be an
instance of $k$-$Synch,$ consider the $k$\textit{-power automaton} $\mathcal{%
M}^{k}$ defined by:

\begin{itemize}
\item $Q_{\mathcal{M}^{k}}=\left\{ \left\{ p_{1},...,p_{k}\right\}
:p_{1},...,p_{k}\in Q_{\mathcal{M}}\right\} $ (the $p$'$s$ are not
necessarily pairwise different).

\item Given $c\in \Sigma _{\mathcal{M}}$, the equality%
\begin{equation*}
\delta _{\mathcal{M}^{k}}\left( c,\left\{ p_{1},...,p_{k}\right\} \right)
=\left\{ \delta _{\mathcal{M}}\left( c,p_{1}\right) ,...,\delta _{\mathcal{M}%
}\left( c,p_{k}\right) \right\}
\end{equation*}%
holds$.$
\end{itemize}

Computing a minimal synchronizing string for states $q_{1},...,q_{r},$ is
the same as computing a minimal path in $\mathcal{M}^{k},$ connecting the
state $\left\{ q_{1},...,q_{r}\right\} $ with the set $\Delta _{\mathcal{M}%
^{k}}=\left\{ A\in Q_{\mathcal{M}^{k}}:\left\vert A\right\vert =1\right\} $.
The later problem can be solved in time $O\left( \left\vert \mathcal{M}%
\right\vert ^{2k}\right) .$
\end{proof}

Notice that, when estimating the running time of the above algorithm, the
parameter $k$ occurs in the exponent. What does it happens if the parameter $%
k$ is not fixed? Let $p$-$Synch\left[ P\right] $ be the parameterized
problem defined by:

\begin{problem}
($p$-\thinspace $synch\left[ P\right] :$ Parameterized synchronization of
planar automata)

\begin{itemize}
\item \textit{Input: }$\left( \mathcal{M},\left\{ q_{1},...,q_{k}\right\}
,k,r\right) ,$ \textit{where\ }$\mathcal{M}$ \textit{is a synchronizing
planar automaton and} \textit{\ }$q_{1},...,q_{k}\in Q_{\mathcal{M}}$\textit{%
.}

\item \textit{Parameter:} $k.$

\item \textit{Problem: Decide if there exists a synchronizing string of
length }$r$ \textit{for the states }$q_{1},...,q_{k}.$
\end{itemize}
\end{problem}

Recall that a parameterized problem is fix parameter tractable, if and only
if, it can be solved in time $O\left( f\left( k\right) \cdot n^{c}\right) $,
for some function $f$ and some constant $c$ (see reference \cite{FG}). Is $p$%
-$Synch\left[ P\right] $ \textit{fix parameter tractable}? We prove that the
problem $p$-$Synch\left[ P\right] $ is $NWL$ complete. The class $NWL$ is
supposed to be the parameterized analogue of\ PSPACE \cite{G}. This class is
located above of the W-hierarchy, and hence our result implies that $p$-$%
Synch\left[ P\right] $ is W$\left[ t\right] $ hard for all $t\geq 1$. \ The
class $NWL$ is defined as the closure under \textit{fpt-reductions\ }(see
reference \cite{FG}) of the following problem.

\begin{problem}
($p$-$NWL:$ deciding acceptance of parameterized space bounded computations)

\begin{itemize}
\item \textit{Input: }$\left( \mathcal{M},t,k\right) $, \textit{where $%
\mathcal{M}$ is a nondeterministic Turing machine, }$t$\textit{\ is a
positive integer given in unary, and} $k\geq 1.$

\item \textit{Parameter:} $k.$

\item \textit{Problem: Decides if} $\mathcal{M}$ \textit{accepts the empty
input in at most} $t$ \textit{steps and checking at most} $k$ \textit{cells}.
\end{itemize}
\end{problem}

Given a parameterized problem $L$, if one wants to check that $L$ belongs to 
$NWL$, it is enough to exhibit a nondeterministic RAM accepting the language 
$L,$ and such that the number of registers it uses along the computation, on
input $X,$ is bounded above by a quantity that only depends on the parameter
of $X$ (see \cite{G}). We prove that $p$-$Synch\left[ P\right] $ is $NWL$
hard by exhibiting an fpt Turing reduction of \textit{The parameterized
longest common subsequence problem} in $p$-$Synch\left[ P\right] .$ The%
\textit{\ }parameterized longest common subsequence problem, denoted by $p$-$%
LCS$, is the parameterized problem defined by:

\begin{problem}
($p$-$LCS:$ parameterized longest common subsequence)
\end{problem}

\begin{itemize}
\item \textit{Input: }$\left( \left\{ w_{1},...,w_{k}\right\} ,\Sigma
,m\right) ,$ \textit{where\ }$\Sigma $ \textit{is a finite alphabet}, $%
w_{1},...,w_{k}\in \Sigma ^{\ast }$\textit{\ and }$m$ \textit{is a positive
integer.}

\item \textit{Parameter:} $k$

\item \textit{Problem: Decide if there exists a string }$w\in \Sigma ^{\ast
} $, \textit{such that for all} $i\leq k$ \textit{string} $w$ \textit{is a
substring of }$w_{i}$, \textit{and such that }$\left\vert w\right\vert =m.$
\end{itemize}

Guillemot \cite{G} proved that $p$-$LCS$ is hard for $NWL$.

\begin{theorem}
\label{teo}The problems $p$-$Synch\left[ P\right] $ and $p$-$Synch$ are $NWL$
complete.
\end{theorem}

\begin{proof}
First, we check that $p$-$Synch$ belongs to $NWL$. To this end, we construct
a suitable nondeterministic RAM accepting the problem $p$-$Synch.$ The
machine works, on input, $\left( \mathcal{M},\left\{ q_{1},...,q_{k}\right\}
,l,k\right) ,$ as follows:

The machine stores in the first $k$ registers a tuple of positive integers $%
\left( s_{1},...,s_{k}\right) $, such that for all $i\leq k$ the inequality $%
s_{i}\leq \left\vert Q\right\vert $ holds. It begins with $\left(
0,...,0\right) $, and then it overwrites $\left( q_{1},...,q_{k}\right) $.
Set $\left( s_{1}^{1},...,s_{k}^{1}\right) =\left( q_{1},...,q_{k}\right) ,$
for all $i\leq l$ the machine nondeterministically chooses a tuple $\left(
s_{1}^{i+1},...,s_{k}^{i+1}\right) $, which can (over)writes on the first $k$
registers, if and only if, there exists $a\in \Sigma $ such that $\delta
\left( a,s_{j}^{i}\right) =s_{j}^{i+1}$. The machine accepts if and only if
the entries of the last tuple are all equal.

Now, we prove the $NWL$ hardness of $p$-$Synch\left[ P\right] $ and $p$-$%
Synch$. First, we prove that $p$-$LCS$ is $fpt$ many-one reducible to $p$-$%
Synch$, and hence we prove that $p$-$Synch$ is $fpt$ Turing reducible to $p$-%
$Synch\left[ P\right] .$ The later reduction is given as the composition of
two reductions. The first one is a $fpt$ many-one reduction of the problem $%
p $-$Synch$ in the problem $p$-$Synch\left[ 2\right] $, which is the
restriction of $p$-$Synch$ to binary automata (automata whose input alphabet
has size 2)$.$ The second one is a $fpt$ many-one reduction of $p$-$Synch%
\left[ 2\right] $ in $p$-$Synch\left[ P\right] .$

\textbf{First stage }(\textit{Reducing }$p$-$LCS$ \textit{to }$p$-$Synch$)%
\textbf{. }

Let $X=\left( \left\{ w_{1},...,w_{k}\right\} ,\Sigma ,m\right) $ be an
instance of $p$-$LCS$. Given $i\leq k,$ we use Baeza-Yates construction\
(see \cite{BY}) to compute a DFA, say $\mathcal{M}_{i}$, that accepts the
language constituted by all the subsequences of $w_{i}.$ It is important to
remark that the size of $\mathcal{M}_{i}$ is bounded above by $\left\vert
w_{i}\right\vert +1$.

Notice that for all $i\leq k,$ we are using the automaton $\mathcal{M}_{i}$
as a language acceptor, it implies that for all $i\leq k$, there exists a
marked state (the initial state of $\mathcal{M}_{i}$) which we denote with
the symbol $q_{0}^{i}.$ Moreover, for all $i\leq k,$ there exists a nonempty
subset of $Q_{i}$, denoted with the symbol $A_{i}$, and which is equal to
the set of accepting states of automaton $\mathcal{M}_{i}$.

We use the set $\left\{ \mathcal{M}_{i}:i\leq k\right\} $ to define an
automaton $\mathcal{M}=\left( \Omega ,Q,\delta \right) $ in the following
way:

\begin{enumerate}
\item $\Omega =\Sigma \cup \left\{ d\right\} $, where $d\notin \Sigma .$

\item $Q=\left( \bigsqcup\limits_{i\leq k}Q_{i}\right) \sqcup \left\{
q,p_{1},...,p_{m+1}\right\} $, where $\sqcup $ denotes disjoint union, and
given $i\leq k,$ the symbol $Q_{i}$ denotes the set of states of the
automaton $\mathcal{M}_{i}.$ Moreover, we have that $q,p_{1},...,p_{m+1}%
\notin \bigsqcup\limits_{i\leq k}Q_{i}$.

\item The transition function of $\mathcal{M}$, which we denote with the
symbol $\delta ,$ is defined as follows%
\begin{equation*}
\delta \left( p,a\right) =\left\{ 
\begin{array}{c}
\delta _{i}\left( p,a\right) \text{, if }p\in Q_{i}\text{ and }a\neq d \\ 
q\text{, if }p\in \bigsqcup\limits_{i\leq k}A_{i}\text{ and }a=d \\ 
p_{1}\text{, if }p\in \left( Q_{i}\backslash A_{i}\right) \text{ and }a=d \\ 
q\text{, if }p=q \\ 
p_{j+1},\text{ if }p=p_{j}\text{, }j<m+1\text{ and }a\in \Sigma \\ 
p_{1}\text{, if }p=p_{j}\text{, }j<m+1\text{, and }a=d \\ 
q,\text{ if }p=p_{m+1}\text{ and }a=d \\ 
p_{1},\text{ if }p=p_{m+1}\text{ and }a\neq d%
\end{array}%
\right.
\end{equation*}
\end{enumerate}

Let $Y\left( X\right) $ be equal to $\left( \mathcal{M},\left\{
q_{0}^{1},...,q_{0}^{k},p_{1}\right\} ,k+1,m+1\right) $, we have that $%
Y\left( X\right) $ it is the output of the first reduction. It is easy to
check that $X\in p$-$LCS$, if and only if, the states $%
q_{0}^{1},...,q_{0}^{k},p_{1}$ can be synchronized in time $m+1,$ that is:
It can be easily checked that $X\in p$-$LCS,$ if and only if, $Y\left(
X\right) \in p$-$Synch$.

Unfortunately, it happens that Baeza-Yates construction is nonplanar, and
hence if $Y\left( X\right) $ is equal to $\left( \mathcal{M},\left\{
q_{0}^{1},...,q_{0}^{k},p_{1}\right\} ,k+1,m+1\right) ,$ it could occur that
the automaton $\mathcal{M}$ is a nonplanar one. Therefore, we have to
proceed with the second reduction.\\

\textbf{Second stage }(\textit{Reducing }$p$-$Synch$ \textit{to }$p$-$Synch%
\left[ P\right] $)

Let $p$-$Synch\left[ 2\right] $ be the restriction of $p$-$Synch$ to the set
of instances%
\begin{equation*}
\left\{ \left( \mathcal{M},\left\{ q_{1},...,q_{k}\right\} ,l,k\right) :%
\mathcal{M}\text{ is binary synchronizing automaton}\right\} .
\end{equation*}

The construction used in \cite{B2} yields a $fpt$ many-one reduction of the
problem $p$-$Synch$ in its restriction $p$-$Synch\left[ 2\right] $. Now, we
will exhibit a $fpt$ Turing reduction of the problem $p$-$Synch\left[ 2%
\right] $ in the problem $p$-$Synch\left[ P\right] $.

Let $\left( \mathcal{M},\left\{ q_{1},...,q_{k}\right\} ,m\right) $ be an
instance of $p$-$Synch\left[ 2\right] .$ A planar drawing of the automaton $%
\mathcal{M}$ is an embedding in $\mathbb{R}^{2}$ of its transition digraph,
and which satisfies the following three properties:

\begin{itemize}
\item Edges are mapped on simple curves.

\item No three edges meet at a common crossing.

\item Two edges meet at most once.
\end{itemize}

Planar drawings can be computed in polynomial time, and if the automaton $%
\mathcal{M}$ is a planar one hence its planar drawing can be chosen to be a
planar embedding. Suppose that $\mathcal{M}$ is nonplanar, and let $\rho $
be a planar drawing of $\mathcal{M}$. Given $e$ an edge (transition) of $%
\mathcal{M}$, we use the symbol $cr_{\rho }\left( e\right) $ to denote the
number of crossings involving edge $e.$ Notice that for all $\rho $ and for
all $e$ the inequality $cr_{\rho }\left( e\right) \leq 2\left\vert
Q\right\vert $ holds.

To begin with the reduction we compute a planar drawing of $\mathcal{M}$,
say $\rho $, and we use $\rho $ to compute a planar automaton $\mathcal{N}$.
The computation of $\mathcal{N}$ goes as follows:

\begin{enumerate}
\item Let $\left\{ a,b\right\} $ be the input alphabet of $\mathcal{M}$, the
input alphabet of $\mathcal{N}$ is equal to $\left\{ a,b\right\} \times
\left\{ 0,1\right\} .$

\item Let $e$ be an edge of $\mathcal{M}$, and suppose that $e$ is labeled
with the letter $a$. We partition $\rho \left( e\right) $ into $2\left\vert
Q\right\vert $ disjoint segments. The idea is to built $2\left\vert
Q\right\vert $\ gadgets that are used to eliminate the crossings involving $%
e $. The segments can be chosen to be connected, with a nonempty interior,
and such that each crossing is an inner point of one of those intervals.
Moreover, we can choose the $2\left\vert Q\right\vert $ segments in such a
way that each one of them contains at most one crossing. The gadgets are
extremely simple:

Suppose that $e$ is directed from $p$ to $q.$ We observe that each one of
the $2\left\vert Q\right\vert $ segments has a first point (the closest to $%
\rho \left( p\right) $)$.$ Given $1\leq i\leq 2\left\vert Q\right\vert $, we
choose four points in the $i$th segment$.$ Let $%
v_{1}^{e,i},v_{2}^{e,i},v_{3}^{e,i}$ and $v_{4}^{e,i}$ be those four points,
we have that $v_{1}^{e,i}$ is equal to the first point of the segment, point 
$v_{2}^{e,i}$ lies between $v_{1}^{e,i}$ and $v_{3}^{e,i}$, while $%
v_{3}^{e,i}$ lies between $v_{2}^{e,i}$ and $v_{4}^{e,i}$. Moreover, the
point $v_{4}^{e,i}$ is different to the first point of the $\left(
i+1\right) $th segment. If $i=1,$ we have that $v_{1}^{e,i}=\rho \left(
p\right) $. If $i=2\left\vert Q\right\vert ,$ we set $v_{1}^{e,i+1}=\rho
\left( q\right) $. Let $i\leq 2\left\vert Q\right\vert $, notice that the $i$%
th segment has been divided into four subsegments $%
e_{1}^{i},e_{2}^{i},e_{3}^{i}$ and $e_{4}^{i}$. Given $j\leq 3$, the edge $%
e_{j}^{i}$ is directed from $v_{j}^{e,i}$ to $v_{j+1}^{e,i}$, while the edge 
$e_{4}^{i},$ is directed from $v_{4}^{e,i}$ to $v_{1}^{e,i+1}$. Moreover, we
assign to those four edges the labels $\left( a,0\right) ,\left( a,1\right)
,\left( a,1\right) $ and $\left( a,0\right) $ (respectively).

\item Now suppose that edges $e$ and $f$ meet at some point $x$. There
exists $i,j\leq 2\left\vert Q\right\vert $ such that $x$ lies on the $i$th
segment of $e,$ and it also lies on the $j$th segment of $f.$ We can choose
the points $v_{1}^{e,i},v_{2}^{e,i},v_{3}^{e,i}$ and $v_{4}^{e,i}$, and the
points $v_{1}^{f,j},v_{2}^{f,j},v_{3}^{f,j}$ and $v_{4}^{f,j}$ in such a way
that:

\begin{itemize}
\item The equalities $v_{3}^{e,i}=v_{2}^{f,j}$ and $v_{4}^{e,i}=v_{3}^{f,j}$
hold.

\item The point $x$ lies between $v_{3}^{e,i}$ and $v_{4}^{e,i}$
\end{itemize}

Notice that the construction is somewhat asymmetrical. However, it does not
matter: Given two edges that cross each other at some point $x,$ it makes
not difference which edge plays the role of edge $e$ and which one plays the
role of edge $f.$

Suppose that $f$ is labeled with the letter $b$, and suppose that $b\neq a$,
in this case the elimination of the crossing looks as follows
\begin{figure}[h]
\begin{center}
\includegraphics[scale =0.8]{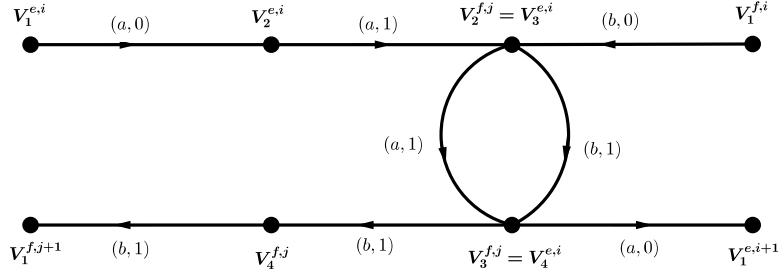}

\end{center}
\caption{Elimination of the crossing }
\label{figura2:Figure 2}
\end{figure}



If $b=a,$ we merge the two parallel edges going from $v_{3}^{e,i}$ ($%
v_{2}^{f,j}$) to $v_{4}^{e,i}$ ($v_{3}^{f,j}$), and which are labeled with
the letter $\left( a,1\right) .$

\item We are trying to draw a planar automaton $\mathcal{N}$. To this end,
we use the set of points%
\begin{equation*}
P=\left\{ v_{j}^{e,i}:i\leq 2\left\vert Q\right\vert ,\text{ }j=1,2,3,4\text{
and }e\text{ is an edge of }\mathcal{M}\right\} ,
\end{equation*}%
and the set of edges%
\begin{equation*}
E=\left\{ e_{j}^{i}:i\leq 2\left\vert Q\right\vert ,\text{ }j=1,2,3,4\text{
and }e\text{ is an edge of }\mathcal{M}\right\} ,
\end{equation*}%
and we add some loops: If $v\in P$, and there is not an outgoing edge
labeled with the letter $x\in \Sigma \times \left\{ 0,1\right\} ,$ we add a
loop with this label attached to this node.
\end{enumerate}

It is easy to check that $\mathcal{N}$ is a planar synchronizing automaton,
which can be computed in polynomial time from automaton $\mathcal{M}$.
Notice that the construction of $\mathcal{N}$ does not depend on the pair $%
\left( \left\{ q_{1},...,q_{k}\right\} ,m\right) .$

Let $m\geq 1$ and let $p$ be an state of $\mathcal{M}$, we use the symbol $%
\mathcal{N}_{m,p}$ to denote the planar automaton that is obtained from $%
\mathcal{N}$ by attaching to node $\rho \left( p\right) $ a planar digraph
that we call $\mathcal{C}_{m}.$ The gadget $\mathcal{C}_{m}$ is computed
from $m$ and $\left\vert Q\right\vert $, and it is used as a clock for the
synchronization process. The construction of $\mathcal{C}_{m}$ goes as
follows:

We begin with a set $W_{m},$ which is equal to 
\begin{equation*}
\left\{ w_{,j,\epsilon }^{i,k}:i\leq 2\left\vert Q\right\vert ,\text{ }k\leq
m,\text{ }j=1,2,3,4\text{, and }\epsilon =0,1\right\} \bigcup \left\{
w\left( p\right) ,\phi \left( p\right) \right\} .
\end{equation*}

Then, for all $k\leq m$ we identify the points $w_{,j,0}^{1,k}$ and $%
w_{,j,1}^{1,k}.$ Given $i\leq 2\left\vert Q\right\vert $ and given $\epsilon
\in \left\{ a,b\right\} ,$ we add the edges $\left( w_{1,\epsilon
}^{i,k},w_{2,\epsilon }^{i,k}\right) $ and $\left( w_{4,\epsilon
}^{i,k},w_{1,\epsilon }^{i+1,k}\right) $ and we label them with the letter $%
\left( \epsilon ,0\right) $. We also add the edges $\left( w_{2,\epsilon
}^{i,k},w_{3,\epsilon }^{i,k}\right) $ and $\left( w_{3,\epsilon
}^{i,k},w_{4,\epsilon }^{i,k}\right) $, and we label them with the letter $%
\left( \epsilon ,1\right) .$ We set $w_{1,0}^{1,m+1}=w_{1,1}^{1,m+1}=\phi
\left( p\right) .$ If $i=2\left\vert Q\right\vert $, we add the edges $%
\left( w_{4,\epsilon }^{i,k},w_{1,\epsilon }^{1,k+1}\right) $ and we label
them with the letter $\left( \epsilon ,0\right) .$ Moreover, for all $%
\epsilon \in \left\{ a,b\right\} $ we add an edge $\left( w\left( p\right)
,w_{1,\epsilon }^{1,1}\right) $ labeled with the letter $\left( \epsilon
,0\right) .$ Finally, we add the necessary loops in order to get a planar
automaton, this planar automaton (its transition digraph) is the clock $%
\mathcal{C}_{m}$.

Given the node $\rho \left( p\right) $, we embed $\mathcal{C}_{m}$ in the
plane in such a way that the following two conditions are satisfied:

\begin{enumerate}
\item $\phi \left( p\right) =\rho \left( p\right) .$

\item There are not crossings between the edges of $\mathcal{C}_{m}$ and the
edges of $\mathcal{N}$.
\end{enumerate}

Then, we remove the loops that were attached to $\phi \left( p\right) $ in
order to get a deterministic planar synchronizing automaton that we denote
with the symbol $\mathcal{N}_{m,p}$.

We observe that if a token is placed on state $w\left( p\right) ,$ then one
can move this token to the state $\rho \left( p\right) $ by using a string
of length $8m\left\vert Q\right\vert .$ The quantity $8m\left\vert
Q\right\vert $ is the length of the shortest strings satisfying the equality 
\begin{equation*}
\widehat{\delta _{\mathcal{N}_{m,p}}}\left( w\left( p\right) ,X\right) =\rho
\left( p\right) .
\end{equation*}

Moreover, if the string $X\in \left( \left\{ a,b\right\} \times \left\{
0,1\right\} \right) ^{8m\left\vert Q\right\vert }$ satisfies the above
equality, there exist $X_{1},...,X_{m}$ such that $X=X_{1}\cdot \cdots \cdot
X_{m}$ and for all $i\leq m$ the factor $X_{i}$ satisfies the equality%
\begin{equation*}
X_{i}=\left( \left( \epsilon _{i},0\right) \left( \epsilon _{i},1\right)
\left( \epsilon _{i},1\right) \left( \epsilon _{i},0\right) \right)
^{2\left\vert Q\right\vert },
\end{equation*}%
for some $\epsilon _{i}\in \left\{ a,b\right\} .$ That is, given $f:\left\{
0,1\right\} ^{\ast }\rightarrow \left( \left( \left\{ 0,1\right\} \times
\left\{ 0,1\right\} \right) \right) ^{\ast }$, the homomorphism defined by%
\begin{equation*}
f\left( a_{1}...a_{k}\right) =\left( \left( a_{1},0\right) \left(
a_{1},1\right) \left( a_{1},1\right) \left( a_{1},0\right) \right)
^{2\left\vert Q\right\vert }....\left( \left( a_{k},0\right) \left(
a_{k},1\right) \left( a_{k},1\right) \left( a_{k},0\right) \right)
^{2\left\vert Q\right\vert },
\end{equation*}%
it happens that $X$ is a minimal string satisfying the equality%
\begin{equation*}
\widehat{\delta _{\mathcal{N}_{m,p}}}\left( w\left( p\right) ,X\right) =\rho
\left( p\right) ,
\end{equation*}%
if and only if, there exists $W_{X}\in \left\{ 0,1\right\} ^{m}$ such that $%
X=f\left( W_{X}\right) .$ This property of $\mathcal{C}_{m}$ allows us to
use it as a clock: If one wants to synchronize the states $%
p_{1},...,p_{k},w\left( p\right) $ in less than $8m\left\vert Q\right\vert
+1 $ steps, then he must try to move all those states to $\rho \left(
p\right) $, and to this end he has to use a string in the range of the
homomorphism $f$. We use this fact to avoid that some token being
synchronized uses the crossing-gadgets to leave an edge of $\mathcal{M}$
that it has not fully traversed, recall that the edges of $\mathcal{M}$ were
partitioned in many different segments (subedges).

Given the automaton $\mathcal{M}$ and given $S=\left\{
q_{1},...,q_{k}\right\} $, we use the symbol $\mathcal{I}_{p,m,S}$ to denote
the tuple%
\begin{equation*}
\left( \mathcal{N}_{m,p},\left\{ \rho \left( q_{1}\right) ,...,\rho \left(
q_{k}\right) ,w\left( p\right) \right\} ,8m\left\vert Q\right\vert \right) .
\end{equation*}

It is easy to check that the set $S$ can be synchronized in time $m$, if and
only if, there exists $p$ such that the states $\rho \left( q_{1}\right)
,...,\rho \left( q_{k}\right) ,w\left( p\right) $ of the automaton $\mathcal{%
N}_{m,p}$ can be synchronized in time $8m\left\vert Q\right\vert $. Thus, we
have the claimed $fpt$ Turing reduction of $p$-$Synch\left[ 2\right] $ in $p$%
-$Synch\left[ P\right] ,$ and hence we have that $p$-$Synch\left[ P\right] $
is $NWL$ hard.

It is worth to remark that the above reduction shows that $p$-$Synch\left[
P,4\right] $ is $NWL$ hard. We use the symbol $p$-$Synch\left[ P,4\right] $
to denote the restriction of $p$-$Synch\left[ P\right] $ to the class of
automata defined over a four letter alphabet.
\end{proof}

Given $L\left[ P\right] $, one of the algorithmic problems studied so far,
we use the symbol $L$ to denote its unrestricted (nonplanar) version, i.e.
symbol $L$ denotes the algorithmic problem that is obtained from $L\left[ P%
\right] $ by flipping the planarity constraint. Let us summarize all the
above results with the following table

\begin{center}
\bigskip 
\begin{equation*}
\begin{tabular}{|l|l|l|l|}
\hline
$Synch\left[ P\right] $ & $Esynch\left[ P\right] $ & Approx. ratio of $Synch%
\left[ P\right] $ & $\ \ p$-$Synch\left[ P\right] $ \\ \hline
NP complete & DP complete & \ \ \ \ \ \ \ \ \ \ $\ \ \ \ O\left( n\right) $
& $NWL$ complete \\ \hline
$Synch$ & $Esynch$ & \ Approx. ratio of $Synch$ & $\ \ \ \ p$-$Synch$ \\ 
\hline
NP complete & DP complete & $\ \ \ \ \ \ \ \ \ \ \ \ \ \ O\left( n\right) $
& $NWL$ complete \\ \hline
\end{tabular}%
\end{equation*}

\bigskip
\end{center}

The above table seems to indicate that the class of planar automata is
universal with respect to the algorithmic complexity of synchronization.
Perhaps, the only issue that remains to be analyzed is the parameterized
approximability of subset synchronization.

We say that $p$-$Synch$ is \textit{fpt approximable within the ratio} $%
f\left( n,k\right) ,$ if and only if, there exists an $fpt$ algorithm,
which, on input $\left( \mathcal{M},\left\{ q_{1},...,q_{k}\right\}
,k\right) $, outputs an integer $t$ such that if $rl_{\mathcal{M}}\left(
q_{1},...,q_{k}\right) $ is the minimal reset length of the states $%
q_{1},...,q_{k},$ then the inequalities%
\begin{equation*}
rst_{\mathcal{M}}\left( q_{1},...,q_{k}\right) \leq t\leq f\left( n,k\right)
\cdot rst_{\mathcal{M}}\left( q_{1},...,q_{k}\right)
\end{equation*}%
hold.

It follows from the work of Gerbush and Heeringa that $p$-$Synch$ is fpt
approximable within the ratio $\left\lceil \frac{n-1}{k-1}\right\rceil $
(see \cite{GH}). It is natural to ask: Which are the approximation ratios
that can achieved in fpt time? Which are the approximation ratios that can
achieved in fpt time for planar automata? We think that those two questions
are the questions that remain to be solved, and that are related to the
algorithmic complexity of synchronizing planar and nonplanar automata.

We observe that subset synchronization makes sense for nonsynchronizing
automata. It is easy to check that for all $k\geq 2$ the synchronizing times
of the hardest $k$-tuples of states is $\Omega \left( n^{k}\right) $ (see 
\cite{V}). Synchronizing times of order $\Omega \left( n^{k}\right) $ are
achieved by sequences of nonsynchronizing automata (an upper bound $O\left(
k\cdot n^{2}\right) $ holds for synchronizing automata). It is also easy to
check that for all $k\geq 2$ there exist sequences of planar automata
achieving those worst synchronizing times of order $\Omega \left(
n^{k}\right) $. Thus, we have that the slowest nonsynchronizing automata are
planar. We conjecture that an analogous fact holds for the synchronization
of all the states (whenever it is possible): The slowest synchronizing
automata are planar automata.

\section{Concluding remarks: Synchronizing times and The \v{C}erny
Conjecture for planar automata}

The hardness of a class of synchronizing automata can be measured in many
different ways, we propose two different hardness measures:

\begin{itemize}
\item The computational hardness of the algorithmic problems (restrictions)
that are determined by the class.

\item The synchronizing times required by the automata within the class.
\end{itemize}

According to the first measure, the class of planar synchronizing automata
is as hard as the class constituted by all the synchronizing automata. We
conjecture that the same is true for the second measure.

\textit{The weak \v{C}erny conjecture} is the conjecture claiming that there
exists a quadratic polynomial $q\left( X\right) $ such that the
synchronizing time of any synchronizing automata with $n$ states is
upperbounded by $q\left( n\right) $. \v{C}erny's conjecture claims that $%
q\left( n\right) $ can be taken equal to $\left( n-1\right) ^{2}.$

\begin{conjecture}
\label{conjecture-universal}Given $\varepsilon >0$, if there exists a
sequence of synchronizing automata whose synchronizing time is $\Omega
\left( n^{2+\varepsilon }\right) $, then there must exist a sequence of
planar synchronizing automata whose synchronizing time is $\Omega \left(
n^{2+\varepsilon }\right) $.
\end{conjecture}

We notice that our conjecture implies that The weak \v{C}erny conjecture is
true, if and only if, it holds true for planar synchronizing automata. In
order to prove the above conjecture one can try a construction similar to
the used in the proof of theorem \ref{teo}.

Let $\mathcal{M}$ be a given nonplanar automaton, and let $\mathcal{N}$ be
the output of the aforementioned construction. We have that $rl_{\mathcal{N}%
}\geq 8n\cdot rl_{\mathcal{M}}$, where $n$ is the size of $\mathcal{M}$.
Thus, the reset length of $\mathcal{N}$ is large provided that the reset
length of $\mathcal{M}$ is large. The problem is that the size of $\mathcal{N%
}$ is quadratic with respect to the size of $\mathcal{M}$. Our construction
does not work because of this quadratic blow-up. It could work if we could
restrict its application to sequences of automata of bounded genus. In this
later case we would have to use smaller clocks, clocks whose sizes are
linearly related to the sizes of the automata given as input.

We notice that from a naive point of view our conjecture must be true:
Planarity is a constraint that makes harder the movement of tokens trough
the digraph. This naive observation was an additional motivation for our
conjecture, which asserts that the worst synchronizing times are achieved by
planar automata. We have some further reasons to consider that it is a
likely conjecture:

\begin{itemize}
\item As remarked before, all the sequences of slowly synchronizing automata
registered in the literature are sequences of planar automata. Thus, the
slowest synchronizing automata registered in the literature are planar
automata.

\item As it was remarked at the end of last section, the slowest
nonsynchronizing automata are planar.
\end{itemize}

We would like to finish this work by proposing two problems:

\begin{itemize}
\item \textbf{Problem 1}: Prove conjecture \ref{conjecture-universal}.

\item \textbf{Problem 2}: Prove The weak \v{C}erny Conjecture for planar
automata.
\end{itemize}

It should be clear that positive solutions to both problems entail a proof
of The weak \v{C}erny Conjecture for general synchronizing automata.

\textbf{Acknowledgement. }The second author would like to thank the support
provided by Universidad Nacional de Colombia through the project Hermes 8943
(32083).

\end{document}